\documentclass[letterpaper, 10 pt, conference]{ieeeconf}  % Comment this line out if you need a4paper

\IEEEoverridecommandlockouts                              % This command is only needed if 
\usepackage{amsmath}
\usepackage{amssymb}
\usepackage{graphicx} 
\usepackage{graphicx}
\usepackage{subcaption}

\usepackage{amsfonts}

\usepackage{amsthm}

\theoremstyle{plain} 
\newtheorem{theorem}{Theorem}

\newtheorem{proposition}{Proposition}

\theoremstyle{definition}
\newtheorem{definition}{Definition}
\newtheorem{assumption}{Assumption}

\theoremstyle{remark}

\title{\LARGE \bf
Performative Scenario Optimization
}

\author{Quanyan Zhu and Zhengye Han% <-注意这里要有一个百分号防空格
\thanks{Q. Zhu and Z. Han are with the Department of Electrical and Computer Engineering, New York University, Brooklyn, NY 11201 USA (e-mail: {\tt\small qz494@nyu.edu}; {\tt\small zh3286@nyu.edu}).}%
}

\begin{document}

\maketitle
\thispagestyle{empty}
\pagestyle{empty}

\begin{abstract}
This paper introduces a performative scenario optimization framework for decision-dependent chance-constrained problems. Unlike classical stochastic optimization, we account for the feedback loop where decisions actively shape the underlying data-generating process. We define performative solutions as self-consistent equilibria and establish their existence using Kakutani’s fixed-point theorem. To ensure computational tractability without requiring an explicit model of the environment, we propose a model-free, scenario-based approximation that alternates between data generation and optimization. Under mild regularity conditions, we prove that a stochastic fixed-point iteration, equipped with a logarithmic sample size schedule, converges almost surely to the unique performative solution. The effectiveness of the proposed framework is demonstrated through an emerging AI safety application: deploying performative guardrails against Large Language Model (LLM) jailbreaks. Numerical results confirm the co-evolution and convergence of the guardrail classifier and the induced adversarial prompt distribution to a stable equilibrium.
\end{abstract}

%%%%%%%%%%%%%%%%%%%%%%%%%%%%%%%%%%%%%%%%%%%%%%%%%%%%%%%%%%%%%%%%%%%%%%%%%%%%%%%%
\section{Introduction}
In many modern decision-making systems, uncertainty is not exogenous, but is instead shaped by the deployed decision itself. As agents adapt, learn, and respond strategically, the data-generating process becomes endogenous and evolves in response to the system. This phenomenon is increasingly prominent in applications such as algorithmic marketplaces, recommendation and ranking systems, autonomous cyber-defense, and large language model (LLM)-driven agent ecosystems. In particular, in security-critical settings, these responses may take the form of evasion or attack, where adversaries probe the system, adapt their strategies, and manipulate inputs to bypass safeguards. Such behaviors have been extensively studied in adversarial machine learning and security through game-theoretic formulations of evasion and poisoning attacks~\cite{hu2024game, zhang2017game, zhang2020security}. Importantly, these adversaries are typically unknown, heterogeneous, and evolving, and therefore cannot be reliably characterized a priori. As a result, the distribution of uncertainty is no longer fixed, but depends on the decision itself, giving rise to a performative feedback loop between the decision and the environment.

Existing approaches to performative optimization address this feedback by modeling the dependence between decisions and the induced distribution. These methods are inherently stochastic and model-based, building on classical stochastic programming and decision-dependent uncertainty frameworks~\cite{shapiro2021lectures}, and more recent formulations of performative prediction and strategic classification~\cite{perdomo2020performative, cohen2024bayesian, mendler2020stochastic, narang2023multiplayer}. They require either an explicit specification or an accurate estimation of how the decision shapes the data-generating process. In doing so, they implicitly rely on assumptions about the agent population, their response mechanisms, and often their learning or equilibrium behavior. In adversarial settings, this further entails modeling attacker objectives and adaptation strategies, often through game-theoretic formulations such as Stackelberg models~\cite{hantoward, liu2025stackelberg}. However, in many modern systems, such assumptions are difficult to justify. Agent behavior may be heterogeneous, nonstationary, and only partially observable, while adversaries may adopt unforeseen or rapidly evolving strategies. Consequently, model-based approaches are fundamentally sensitive to misspecification and may fail to capture the true feedback dynamics present in practice.

In contrast, we adopt a performative scenario optimization perspective that avoids explicit modeling of the induced distribution. Rather than characterizing the mapping from decisions to distributions, we treat the environment as a black-box sampling mechanism and construct decisions directly from data generated under the current deployment. This approach builds on scenario-based and randomized methods for uncertain systems~\cite{calafiore2002randomized, formentin2014scenario, tempo1996probabilistic, tempo2005randomized}. In this sense, the approach is both model-free and distribution-free in an operational sense: it does not require knowledge of the data-generating process, nor of the agent or population model that produces it, including adversarial responses or learning dynamics. Instead, the effects of unknown and adaptive agents are incorporated implicitly through the observed samples. This leads to a tractable scenario-based approximation of the performative problem, where feasibility is enforced empirically using data, and theoretical guarantees depend only on the number of samples and the dimension of the decision space.

This perspective is particularly well-suited for adversarial and evolving environments. As the environment evolves, new samples reflect the updated responses of agents or adversaries, enabling the decision to adapt accordingly. In this way, the framework supports an online and iterative procedure in which optimization and data generation are coupled, while remaining robust to unknown, heterogeneous, and dynamically changing behaviors.

Specifically, the main contributions of this paper are threefold:
\begin{itemize}
    \item \textbf{Novel Performative Framework:} We introduce a performative scenario optimization framework for chance-constrained problems and formally establish the existence of self-consistent equilibria (performative solutions) utilizing Kakutani’s fixed-point theorem.
    \item \textbf{Tractable Algorithmic Design:} To overcome the intractability of unknown induced distributions, we propose a model-free, scenario-based approximation algorithm that alternates between data generation and optimization, effectively decoupling the performative feedback loop into computable steps.
    \item \textbf{Convergence Guarantees and Validation:} Under mild regularity conditions, we prove that our stochastic fixed-point iteration, equipped with a logarithmic sample size schedule, converges almost surely to the unique performative solution. We demonstrate its practical effectiveness through a performative guardrail deployment problem in LLM safety.
\end{itemize}

\section{PRELIMINARIES AND PROBLEM STATEMENT}

\subsection{Decision-Dependent Chance-Constrained Problem}
Let $(\Xi, \mathcal{F})$ be a measurable space representing uncertainty, and let $\xi \in \Xi$ denote a realization of the uncertain environment. Let $X \subset \mathbb{R}^{d}$ be a nonempty compact set of admissible decisions. When needed, we additionally assume that $X$ is convex.
We consider a measurable function $g: X \times \Xi \rightarrow \mathbb{R}$, where $g(x, \xi)$ represents the performance level (or constraint value) associated with decision $x \in X$ under realization $\xi \in \Xi$. Fix a threshold $\gamma \in \mathbb{R}$ and a tolerance level $\varepsilon \in(0,1)$.
For each $x \in X$, let $\mathbb{P}_{x}$ be a probability measure on $(\Xi, \mathcal{F})$. The key modeling feature is that the probability law governing uncertainty depends on the decision $x$.

\begin{assumption}[Regularity]
\label{am:Regularity}
The following conditions hold:\\
(A1) The objective function $f: X \rightarrow \mathbb{R}$ is continuous. When needed, $f$ is assumed convex.\\
(A2) The mapping $g: X \times \Xi \rightarrow \mathbb{R}$ is jointly measurable.\\
(A3) For each $\xi \in \Xi$, the mapping $x \mapsto g(x, \xi)$ is continuous on $X$. When needed, it is convex.\\
(A4) For each $x \in X, \mathbb{P}_{x}$ is a probability measure on $(\Xi, \mathcal{F})$.\\
(A5) The mapping $x \mapsto \mathbb{P}_{x}$ is weakly continuous, i.e., for every bounded continuous function 
$\varphi: \Xi \rightarrow \mathbb{R}$, the map $x \mapsto \int_{\Xi} \varphi(\xi) d \mathbb{P}_{x}(\xi)$ is continuous on $X$.     
\end{assumption}

For a given decision $x \in X$, the probability that the performance satisfies the threshold requirement is $\mathbb{P}_{x}(g(x, \xi) \leq \gamma)$. Accordingly, we say that $x$ is chance-feasible if $\mathbb{P}_{x}(g(x, \xi) \leq \gamma) \geq 1-\varepsilon $, or equivalently $\mathbb{P}_{x}(g(x, \xi)>\gamma) \leq \varepsilon$.
We define the decision-dependent chance-feasible set by $X_{\varepsilon}^{\gamma}:=\left\{x \in X: \mathbb{P}_{x}(g(x, \xi) \leq \gamma) \geq 1-\varepsilon\right\}.$ The corresponding population formulation is
\begin{equation}
\label{population formulation}
\min _{x \in X} f(x) \quad \text { subject to } \quad \mathbb{P}_{x}(g(x, \xi) \leq \gamma) \geq 1-\varepsilon 
\end{equation}

However, \eqref{population formulation} is not a standard optimization problem, because the feasible set depends on the decision variable itself through the mapping $x \mapsto \mathbb{P}_{x}$. As a result, feasibility and optimality must be interpreted jointly.

To make this dependence explicit, for each reference decision $x \in X$ we define the decision dependent feasibility correspondence
\begin{equation}
\label{decision dependent feasibility correspondence}
\mathcal{X}_{\varepsilon}^{\gamma}(x):=\left\{y \in X: \mathbb{P}_{x}(g(y, \xi) \leq \gamma) \geq 1-\varepsilon\right\} 
\end{equation}

Thus, $\mathcal{X}_{\varepsilon}^{\gamma}(x)$ consists of all decisions $y$ that are feasible when performance is evaluated under the probability law induced by $x$.

\begin{definition}[Performative chance-constrained problem]
\label{Performative_chance-constrained_problem}
The performative chance-constrained problem consists of finding a decision $x^{\star} \in X$ such that
\begin{equation}
\label{eq:Performative chance-constrained problem}
x^{\star} \in \arg \min _{y \in \mathcal{X}_{\varepsilon}^{\gamma}\left(x^{\star}\right)} f(y)
\end{equation}
A point $x^{\star}$ satisfying (\ref{eq:Performative chance-constrained problem}) is called a performative solution. Equivalently, $x^{\star}$ satisfies: (i) feasibility: $\mathbb{P}_{x^{\star}}\left(g\left(x^{\star}, \xi\right) \leq \gamma\right) \geq 1-\varepsilon$, and (ii) optimality: $f\left(x^{\star}\right) \leq f(y)$ for all $y \in \mathcal{X}_{\varepsilon}^{\gamma}\left(x^{\star}\right)$.
\end{definition}

This formulation highlights a self-consistency requirement: the decision $x^{\star}$ determines the probability measure $\mathbb{P}_{x^{\star}}$, which in turn defines the feasible set $\mathcal{X}_{\varepsilon}^{\gamma}\left(x^{\star}\right)$ over which optimality is evaluated. Thus, a performative solution must be optimal with respect to a feasible region that it itself induces. This distinguishes the performative setting from classical chance-constrained optimization, where the probability measure is fixed independently of the decision variable.

\subsection{Performative Solution and Best-Response Mapping}
We build on the performative chance-constrained problem defined in Definition~\ref{Performative_chance-constrained_problem}. In this setting, feasibility depends on the probability measure induced by the decision itself, and therefore optimality must be evaluated relative to an endogenous feasible set.

To analyze this structure, we adopt a two-stage viewpoint. For a given reference decision $x \in X$, we evaluate the feasibility of another decision $y \in X$ under the probability law $\mathbb{P}_{x}$. This yields the decision-dependent feasible set $\mathcal{X}_{\varepsilon}^{\gamma}(x)$ defined in \eqref{decision dependent feasibility correspondence}.

\begin{definition}[Best-response operator]
The set-valued mapping $\Phi: X \rightrightarrows X$ is defined by 

\begin{equation}
\label{eq:Best-response operator}
\Phi(x):=\arg \min _{y \in \mathcal{X}_{\varepsilon}^{\gamma}(x)} f(y).    
\end{equation}

For each $x \in X$, the set $\Phi(x)$ consists of all optimal decisions when performance is evaluated under the environment generated by $x$.    
\end{definition}

\begin{definition}[Performative solution]
 A point $x^{\star} \in X$ is called a performative solution if it satisfies
\begin{equation}
\label{Performative solution}
x^{\star} \in \Phi\left(x^{\star}\right)
\end{equation}

\end{definition}

By Definition ~\ref{Performative_chance-constrained_problem}, \eqref{Performative solution} is precisely the solution condition of the performative chance-constrained problem. Equivalently, $x^{\star}$ satisfies both: (i) optimality: $x^{\star} \in \arg \min _{y \in \mathcal{X}_{\varepsilon}^{\gamma}\left(x^{\star}\right)} f(y)$, and (ii) feasibility: $\mathbb{P}_{x^{\star}}\left(g\left(x^{\star}, \xi\right) \leq \gamma\right) \geq 1-\varepsilon$.

Thus, a performative solution is a decision that is optimal with respect to a feasible region that it itself induces. This self-consistency distinguishes the performative setting from classical chance-constrained optimization, where the probability measure is fixed independently of the decision variable. We now study structural properties of $\Phi$.

\begin{assumption}[Boundary regularity]
\label{am:Boundary regularity}
For every $(x, y) \in X \times X$, we have $\mathbb{P}_{x}(g(y, \xi)=\gamma)=0$.
This assumption excludes probability mass at the threshold and ensures stability of feasibility under limits.
\end{assumption}

\begin{proposition}[Compactness of feasible sets and best responses]
\label{pro:Compactness of feasible sets and best responses}
 Suppose that Assumptions \ref{am:Regularity} and \ref{am:Boundary regularity} hold and $\mathcal{X}_{\varepsilon}^{\gamma}(x) \neq \varnothing$. Then $\mathcal{X}_{\varepsilon}^{\gamma}(x)$ is compact and $\Phi(x)$ is nonempty and compact.      
\end{proposition}

\begin{proof}
Since $X$ is compact, it suffices to show that $\mathcal{X}_{\varepsilon}^{\gamma}(x)$ is closed. Let $y_{n} \rightarrow y$ with $y_{n} \in \mathcal{X}_{\varepsilon}^{\gamma}(x)$, so $\mathbb{P}_{x}\left(g\left(y_{n}, \xi\right) \leq \gamma\right) \geq 1-\varepsilon$ \text{ for all} $n$. For each $\xi$, continuity of $y \mapsto g(y, \xi)$ implies $g\left(y_{n}, \xi\right) \rightarrow g(y, \xi)$. Under Assumption \ref{am:Boundary regularity}, the indicators $\mathbf{1}_{\left\{g\left(y_{n}, \xi\right) \leq \gamma\right\}}$ converge almost surely to $\mathbf{1}_{\{g(y, \xi) \leq \gamma\}}$. Applying Fatou's lemma to $1-\mathbf{1}_{\left\{g\left(y_{n}, \xi\right) \leq \gamma\right\}}$, or equivalently using the upper semicontinuity induced by the boundary condition, yields $\mathbb{P}_{x}(g(y, \xi) \leq \gamma) \geq \limsup _{n \rightarrow \infty} \mathbb{P}_{x}\left(g\left(y_{n}, \xi\right) \leq \gamma\right) \geq 1-\varepsilon$, so $y \in \mathcal{X}_{\varepsilon}^{\gamma}(x)$.

Thus $\mathcal{X}_{\varepsilon}^{\gamma}(x)$ is closed and therefore compact. Since $f$ is continuous, it attains its minimum on this set, and the set of minimizers is compact.    
\end{proof}

\begin{proposition}[Closed graph and upper hemicontinuity]
\label{pro:Closed graph and upper hemicontinuity}
 Suppose that Assumptions \ref{am:Regularity} and \ref{am:Boundary regularity} hold. Then $\Phi$ has a closed graph and is upper hemicontinuous.    
\end{proposition}

\begin{proof}
Let $x_{n} \rightarrow x$ and $y_{n} \rightarrow y$ with $y_{n} \in \Phi\left(x_{n}\right)$. Since $y_{n} \in \mathcal{X}_{\varepsilon}^{\gamma}\left(x_{n}\right)$, we have $\mathbb{P}_{x_{n}}\left(g\left(y_{n}, \xi\right) \leq \gamma\right ) \geq 1-\varepsilon$. By weak continuity of $x \mapsto \mathbb{P}_{x}$ and continuity of $g$, together with Assumption \ref{am:Boundary regularity}, it follows that $\mathbb{P}_{x}(g(y, \xi) \leq \gamma) \geq 1-\varepsilon$, hence $y \in \mathcal{X}_{\varepsilon}^{\gamma}(x)$.

Let $z \in \mathcal{X}_{\varepsilon}^{\gamma}(x)$. One can construct $z_{n} \rightarrow z$ with $z_{n} \in \mathcal{X}_{\varepsilon}^{\gamma}\left(x_{n}\right)$ (by continuity and boundary regularity). Since $y_{n}$ minimizes $f$ over $\mathcal{X}_{\varepsilon}^{\gamma}\left(x_{n}\right)$, we have $f\left(y_{n}\right) \leq f\left(z_{n}\right)$. Passing to the limit yields $f(y) \leq f(z)$, so $y \in \Phi(x)$.
Thus $\Phi$ has a closed graph. Since it has nonempty compact values, it is upper hemicontinuous.    
\end{proof}

\begin{theorem}[Existence of performative solution]
Suppose that $X$ is nonempty, compact, and convex, $\mathcal{X}_{\varepsilon}^{\gamma}(x) \neq \varnothing$ for all $x \in X$, and Assumptions \ref{am:Regularity} and \ref{am:Boundary regularity} hold. Then there exists $x^{\star} \in X$ such that
\begin{equation}
x^{\star} \in \Phi\left(x^{\star}\right) . 
\end{equation}   
\end{theorem}

\begin{proof}
By Proposition \ref{pro:Compactness of feasible sets and best responses}, $\Phi(x)$ is nonempty and compact for all $x$. By Proposition \ref{pro:Closed graph and upper hemicontinuity}, $\Phi$ is upper hemicontinuous. Since $\Phi(x) \subseteq X$ and $X$ is compact and convex, Kakutani's fixed-point theorem implies existence of $x^{\star}$ satisfying (1.8).    
\end{proof}

\subsection{Scenario Approach}
The best-response operator $\Phi$ defined in (\ref{eq:Best-response operator}) requires solving $\Phi(x)=\arg \min _{y \in \mathcal{X}_{\varepsilon}^{\gamma}(x)} f(y)$, where $\mathcal{X}_{\varepsilon}^{\gamma}(x)$ is defined through the chance constraint $\mathbb{P}_{x}(g(y, \xi) \leq \gamma) \geq 1-\varepsilon$. In general, this constraint is intractable, as it requires evaluating probabilities under the distribution $\mathbb{P}_{x}$. To obtain a tractable approximation, we adopt a scenario-based approach.

\begin{definition}[Scenario approximation of the best-response problem]
 Fix $x \in X$ and let $\xi^{(1)}, \ldots, \xi^{(N)}$ be i.i.d. samples drawn from $\mathbb{P}_{x}$. The scenario feasible set is defined as $\mathcal{X}_{N}(x)=\bigcap_{i=1}^{N}\left\{y \in X: g\left(y, \xi^{(i)}\right) \leq \gamma\right\}$. 
The corresponding scenario optimization problem is given by
\begin{equation}
\label{eq:scenario_problem}
\begin{aligned}
\min_{y \in X} \quad & f(y) \\
\text{subject to} \quad & g\left(y, \xi^{(i)}\right) \leq \gamma, \quad \forall i=1, \ldots, N.
\end{aligned}
\end{equation}

We define the empirical best-response operator by $\Phi_{N}(x)=\arg \min _{y \in \mathcal{X}_{N}(x)} f(y)$.
For each fixed $x, \Phi_{N}(x)$ is a random set induced by the samples $\left\{\xi^{(i)}\right\}_{i=1}^{N}$. Thus, $\Phi_{N}$ provides a data-driven approximation of the best-response mapping $\Phi$.    
\end{definition}

To distinguish the sources of randomness, we use $\mathbb{P}_{x}$ to denote the probability measure governing the uncertainty $\xi$, and $\mathbb{P}_{\xi^{(1: N)}}$ to denote the probability measure induced by the i.i.d. sampling of the scenarios.
Define the violation probability of a decision $y \in X$ under $\mathbb{P}_{x}$ as $V_{x}(y):=\mathbb{P}_{x}(g(y, \xi)>\gamma)$. We now establish finite-sample guarantees for the scenario solution.

\begin{theorem}[Scenario feasibility guarantee]
 Fix $x \in X$ and assume: (i) $f$ and $g(\cdot, \xi)$ are convex in $y$, (ii) the scenario problem (\ref{eq:scenario_problem}) admits a unique optimizer $y_{N}$ with probability one. Let $d$ denote the dimension of the decision variable. Then, for any $\varepsilon, \beta \in(0,1)$, if $N$ satisfies
\begin{equation}
\label{eq:Scenario feasibility guarantee}
\sum_{i=0}^{d-1}\binom{N}{i} \varepsilon^{i}(1-\varepsilon)^{N-i} \leq \beta
\end{equation}
we have $\mathbb{P}_{\xi^{(1: N)}}\left(V_{x}\left(y_{N}\right) \leq \varepsilon\right) \geq 1-\beta .$   
\end{theorem}

\begin{proof}
Fix $x$ and consider the random samples $\left\{\xi^{(i)}\right\}_{i=1}^{N}$. Let $V_{x}(y)=\mathbb{P}_{x}(g(y, \xi)>\gamma)$. We bound the probability that the scenario solution violates the chance constraint, i.e., $\mathbb{P}_{\xi^{(1: N)}}\left(V_{x}\left(y_{N}\right)>\right. \varepsilon)$. If $V_{x}\left(y_{N}\right)>\varepsilon$, then for each sampled constraint, $\mathbb{P}_{x}\left(g\left(y_{N}, \xi^{(i)}\right) \leq \gamma\right) \leq 1-\varepsilon$. Hence, the probability that $y_{N}$ satisfies all $N$ sampled constraints is at most $(1-\varepsilon)^{N}$.

The key structural property of convex scenario programs is that the optimizer $y_{N}$ is determined by at most $d$ support constraints. Therefore, one can enumerate all possible subsets of active constraints of size at most $d-1$. Applying a union bound over these subsets yields $\mathbb{P}_{\xi^{(1: N)}}\left(V_{x}\left(y_{N}\right)>\varepsilon\right) \leq \sum_{i=0}^{d-1}\binom{N}{i} \varepsilon^{i}(1-\varepsilon)^{N-i}$.
Thus, if (\ref{eq:Scenario feasibility guarantee}) holds, then $\mathbb{P}_{\xi^{(1: N)}}\left(V_{x}\left(y_{N}\right)>\varepsilon\right) \leq \beta$, which proves the result.    
\end{proof}

\section{Fixed-Point Iteration \& Computational Schemes}
\subsection{Deterministic Best-Response Iteration}
We study when the induced best-response map $\Phi$ is a contraction, which yields a constructive algorithm.

\begin{definition}[Deterministic best-response iteration]
Given $x_{0} \in X$, define $\left\{x_{t}\right\}$ by $x_{t+1} \in \Phi\left(x_{t}\right)$.
\end{definition}
To obtain contraction, we impose structure on $f, g$, and the dependence $x \mapsto \mathbb{P}_{x}$.

\begin{assumption}[Strong convexity and smoothness of the objective]
\label{Strong convexity and smoothness of the objective}
The function $f: X \rightarrow \mathbb{R}$ is $\mu$-strongly convex and has $L_{f}$-Lipschitz gradient on $X$ (with $\mu>0$ ).    
\end{assumption}

\begin{assumption}[Regularity of the constraint map]
For every $\xi \in \Xi$, the map $y \mapsto g(y, \xi)$ is convex and $L_{g}$-Lipschitz on $X$; moreover, for each $y$, the function $\xi \mapsto g(y, \xi)$ is measurable and bounded.   
\end{assumption}

\begin{assumption}[Lipschitz dependence of the distribution]
 There exists $L_{P} \geq 0$ such that for all $x, x^{\prime} \in X$ and all $y \in X$,
\begin{equation}
\label{eq:Lipschitz dependence of the distribution}
\left|\mathbb{P}_{x}(g(y, \xi) \leq \gamma)-\mathbb{P}_{x^{\prime}}(g(y, \xi) \leq \gamma)\right| \leq L_{P}\left\|x-x^{\prime}\right\| . 
\end{equation}    
\end{assumption}

\begin{assumption}[Constraint qualification]
\label{Constraint qualification}
For each $x$, the feasible set $\mathcal{X}_{\varepsilon}^{\gamma}(x)$ is nonempty and there exists a unique minimizer $\phi(x):=\arg \min _{y \in \mathcal{X}_{\varepsilon}^{\gamma}(x)} f(y)$.
Moreover, a Slater-type condition holds uniformly: there exists $\bar{y} \in X$ such that $\inf _{x \in X} \mathbb{P}_{x}(g(\bar{y}, \xi) \leq \gamma) \geq 1-\varepsilon+\sigma$ for some $\sigma>0$.
\end{assumption}

Under Assumptions \ref{Strong convexity and smoothness of the objective} to \ref{Constraint qualification}, the best-response is single-valued $\Phi(x)=\{\phi(x)\}$.

\begin{proposition}[Lipschitz continuity of the best response]
\label{Lipschitz continuity of the best response}
Under Assumptions \ref{Strong convexity and smoothness of the objective} to \ref{Constraint qualification}, there exists a constant $K>0$ such that
\begin{equation}
\label{eq:Lipschitz continuity of the best response}
\left\|\phi(x)-\phi\left(x^{\prime}\right)\right\| \leq K\left\|x-x^{\prime}\right\| \quad \text { for all } x, x^{\prime} \in X 
\end{equation}
with $K$ depending only on $\left(\mu, L_{f}, L_{g}, L_{P}, \sigma\right)$. 
\end{proposition}

\begin{proof}
Define the value function $\psi(x):=\min _{y \in X}\left\{f(y): \mathbb{P}_{x}(g(y, \xi) \leq \gamma) \geq 1-\varepsilon\right\}$.
By Assumption \ref{Strong convexity and smoothness of the objective}, $f$ is strongly convex, hence the minimizer $\phi(x)$ is unique.
We show that the feasible sets vary Lipschitzly in $x$. Fix $x, x^{\prime}$ and any $y \in X$. By (\ref{eq:Lipschitz dependence of the distribution}), $\mathbb{P}_{x^{\prime}}(g(y, \xi) \leq \gamma) \geq \mathbb{P}_{x}(g(y, \xi) \leq \gamma)-L_{P}\left\|x-x^{\prime}\right\|$. Hence, if $y \in \mathcal{X}_{\varepsilon}^{\gamma}(x)$, then $\mathbb{P}_{x^{\prime}}(g(y, \xi) \leq \gamma) \geq 1-\varepsilon-L_{P}\left\|x-x^{\prime}\right\| .$
Using the uniform Slater condition (Assumption \ref{Constraint qualification}) and standard constraint-perturbation arguments for convex programs, one can construct for each feasible $y$ a nearby point $\tilde{y}$ feasible for $x^{\prime}$ with $\|\tilde{y}-y\| \leq c_{1} L_{P}\|x-x^{\prime}\|$ for some constant $c_{1}$ depending on $(L_{g}, \sigma)$. 

Now compare optimal solutions. Let $y=\phi(x)$ and $y^{\prime}=\phi(x^{\prime})$. Using strong convexity of $f$, we have $\frac{\mu}{2}\|y-y^{\prime}\|^{2} \leq \langle\nabla f(y^{\prime})-\nabla f(y), y-y^{\prime}\rangle$. By optimality of $y$ and $y^{\prime}$ over their respective feasible sets and the feasibility-transfer argument above, one obtains $f(y^{\prime})-f(y) \leq c_{2}\|x-x^{\prime}\|$ and $f(y)-f(y^{\prime}) \leq c_{2}\|x-x^{\prime}\|$. Combining these with strong convexity yields $\|y-y^{\prime}\| \leq \frac{2 c_{2}}{\mu}\|x-x^{\prime}\|$, which proves (\ref{eq:Lipschitz continuity of the best response}).
\end{proof}

\begin{theorem}[Contraction and convergence]
\label{Contraction and convergence}
Under Assumptions \ref{Strong convexity and smoothness of the objective} to \ref{Constraint qualification}, suppose the Lipschitz constant $K$ in (\ref{eq:Lipschitz continuity of the best response}) satisfies $K<1$. Then, (i) $\phi$ is a contraction on $X$ and admits a unique fixed point $x^{\star}$, and (ii) for any $x_{0} \in X$, the iteration $x_{t+1}=\phi(x_{t})$ converges linearly to $x^{\star}$ satisfying $\|x_{t}-x^{\star}\| \leq K^{t}\|x_{0}-x^{\star}\|$.
\end{theorem}

\begin{proof}
By Proposition \ref{Lipschitz continuity of the best response}, $\phi$ is K-Lipschitz. If $K<1$, it is a contraction on the complete metric space $X \subset \mathbb{R}^{d}$. The claims follow from Banach's fixed-point theorem.    
\end{proof}

The contraction modulus $K$ increases with: (i) the sensitivity $L_{P}$ of the distribution $\mathbb{P}_{x}$ to $x$, (ii) the steepness $L_{g}$ of the constraint, and decreases with: (iii) the strong convexity $\mu$ of $f$ and the constraint margin $\sigma$. Hence, contraction holds when the decision has a weak effect on the induced distribution (small $L_{P}$ ) and the objective is sufficiently well-conditioned (large $\mu$ ).

\subsection{Scenario-Based (Stochastic) Fixed-Point Iteration}

We now build a stochastic counterpart to the deterministic iteration. Recall from Theorem \ref{Contraction and convergence} that, under Assumptions \ref{Strong convexity and smoothness of the objective} to \ref{Constraint qualification}, the best-response map $\phi$ is a contraction with modulus $K<1$ and admits a unique fixed point $x^{\star}$.

Because the exact map $\phi$ is generally intractable, we approximate it using the scenario-based operator, which naturally induces a stochastic iteration. Formally, let $(\Omega, \mathcal{F}, \mathbb{P})$ be a probability space equipped with the filtration $\mathcal{F}_{t}:=\sigma\left(x_{0}, \xi_{0}^{(1: N_{0})}, \ldots, \xi_{t-1}^{(1: N_{t-1})}\right)$. At each step $t$, conditional on $x_{t}$, we draw i.i.d. samples $\{\xi_{t}^{(i)}\}_{i=1}^{N_{t}}$ from $\mathbb{P}_{x_{t}}$ and update the decision variable via:
\begin{equation}
\label{eq:stochastic_iteration}
x_{t+1}=\phi_{N_{t}}\left(x_{t} ; \xi_{t}^{(1: N_{t})}\right).
\end{equation}
To analyze its convergence, we introduce an error decomposition by defining the approximation error as $e_{t}:=\phi_{N_{t}}\left(x_{t} ; \xi_{t}^{(1: N_{t})}\right)-\phi\left(x_{t}\right)$. This allows us to express the stochastic recursion equivalently as $x_{t+1}=\phi\left(x_{t}\right)+e_{t}$.

\begin{assumption}[Uniform scenario approximation]
\label{Uniform scenario approximation}
There exist deterministic sequences $\left\{\eta_{N}\right\}$ and $\left\{\delta_{N}\right\}$ with $\eta_{N} \downarrow 0$ and $\delta_{N} \downarrow 0$ such that, for all $x \in X$,

\begin{equation}
\mathbb{P}\left(\left\|\phi_{N}\left(x ; \xi^{(1: N)}\right)-\phi(x)\right\|>\eta_{N}\right) \leq \delta_{N} .
\end{equation}    
\end{assumption}

Assumption \ref{Uniform scenario approximation} requires that the scenario-based optimizer $\phi_{N}\left(x ; \xi^{(1: N)}\right)$ approximates the exact optimizer $\phi(x)$ uniformly over $x \in X$, with high probability. The quantity $\eta_{N}$ controls the magnitude of the approximation error, while $\delta_{N}$ controls the probability of a large deviation event.
This assumption is justified by combining two ingredients: (i) Scenario feasibility guarantees, which ensure that the scenario solution satisfies the chance constraint with probability at least $1-\delta_{N}$; (ii) Stability of the optimization problem, which follows from strong convexity of $f$ and regularity of the feasible set, and implies that small perturbations in the constraint set lead to small perturbations in the optimizer.

Under Assumptions \ref{Strong convexity and smoothness of the objective} to \ref{Constraint qualification}, standard sensitivity analysis for strongly convex programs yields $\eta_{N}=O\left(\sqrt{\frac{d}{N}}\right)$. Moreover, from the scenario optimization bound, one can choose $\delta_{N}$ such that $\delta_{N} \leq \exp (-c N \varepsilon)$ for some constant $c>0$.

\begin{assumption}[Sample schedule]
\label{Sample schedule}

The sample sizes $\left\{N_{t}\right\}$ are chosen so that $\sum_{t=0}^{\infty} \delta_{N_{t}}<\infty \text { and } \eta_{N_{t}} \rightarrow 0 $
 
\end{assumption}

Assumption \ref{Sample schedule} ensures that the stochastic approximation errors vanish almost surely. The summability condition $\sum_{t} \delta_{N_{t}}<\infty$ implies, by the Borel-Cantelli lemma, that only finitely many iterations incur large approximation errors. The condition $\eta_{N_{t}} \rightarrow 0$ ensures that the magnitude of the approximation error converges to zero.
It can be verified that a logarithmic sample size schedule, $N_{t}=\mathcal{O}(d+\ln t)$, is sufficient to satisfy the vanishing error conditions in Assumption \ref{Sample schedule}. The detailed derivation is deferred to the extended version. Building upon this, we establish the convergence of the stochastic fixed-point iteration.

\begin{theorem}[Almost sure convergence of stochastic iteration]
\label{Almost sure convergence of stochastic iteration}
Under Assumptions \ref{Strong convexity and smoothness of the objective} to \ref{Constraint qualification} and Assumptions \ref{Uniform scenario approximation} to \ref{Sample schedule}, the sequence $\left\{x_{t}\right\}$  converges $\mathbb{P}$-almost surely to the unique performative solution $x^{\star}$.    
\end{theorem}

\begin{proof}
 By Theorem \ref{Contraction and convergence}, the map $\phi$ is a contraction with modulus $K<1$ and admits a unique fixed point $x^{\star}$. Define the error $e_{t}:=x_{t+1}-\phi\left(x_{t}\right)$ and the events $B_{t}:=\left\{\left\|e_{t}\right\|>\eta_{N_{t}}\right\} .$
Since the samples $\xi_{t}^{\left(1: N_{t}\right)}$ are drawn conditionally on $x_{t}$, the conditional probability is taken with respect to $\mathbb{P}_{x_{t}}^{\otimes N_{t}}$. Hence, by Assumption \ref{Uniform scenario approximation}, $\mathbb{P}\left(B_{t} \mid \mathcal{F}_{t}\right) \leq \delta_{N_{t}} \text { a.s. }$

Taking expectations yields $\mathbb{P}\left(B_{t}\right) \leq \delta_{N_{t}}$. Since $\sum_{t=0}^{\infty} \delta_{N_{t}}<\infty$, the Borel-Cantelli lemma implies $\mathbb{P}\left(B_{t} \text { infinitely often }\right)=0$.
Thus, with probability one, there exists a random time $T(\omega)$ such that for all $t \geq T(\omega)$, $\left\|e_{t}\right\| \leq \eta_{N_{t}} .$ Fix $\omega$ in this probability-one event. For all $t \geq T(\omega)$, we have $\left\|x_{t+1}-x^{\star}\right\| \leq\left\|\phi\left(x_{t}\right)-\phi\left(x^{\star}\right)\right\|+\left\|e_{t}\right\| \leq K\left\|x_{t}-x^{\star}\right\|+\eta_{N_{t}} .$
Define $a_{t} := \|x_{t}-x^{\star}\|$. Then, for all $t \geq T(\omega)$, we have $a_{t+1} \leq K a_{t}+\eta_{N_{t}}$. Iterating this inequality from $T:=T(\omega)$ yields $a_{t} \leq K^{t-T} a_{T} + \sum_{s=T}^{t-1} K^{t-1-s} \eta_{N_{s}}.$

Since $K \in(0,1)$, the first term $K^{t-T} a_{T} \rightarrow 0$ as $t \rightarrow \infty$. To bound the summation, fix any $\varepsilon>0$. Because $\eta_{N_{s}} \rightarrow 0$, there exists an integer $S \geq T$ such that $\eta_{N_{s}} \leq \varepsilon(1-K)$ for all $s \geq S$. Consequently, for all $t \geq S+1$, the tail of the sum satisfies $\sum_{s=S}^{t-1} K^{t-1-s} \eta_{N_{s}} \leq \varepsilon(1-K) \sum_{k=0}^{t-1-S} K^{k} \leq \varepsilon.$
The remaining finitely many terms for $s=T, \ldots, S-1$ vanish as $t \rightarrow \infty$ due to geometric decay. Hence, $\limsup _{t \rightarrow \infty} a_{t} \leq \varepsilon$. Since $\varepsilon>0$ is chosen arbitrarily, we conclude that $a_{t} \rightarrow 0$, which implies $x_{t} \rightarrow x^{\star}$ $\mathbb{P}$-almost surely.
\end{proof}

\section{Equivalent Game-Theoretic Model}
\label{sec:Game-Theoretic}
In the preceding sections, the decision-dependent distribution mapping $x \mapsto \mathbb{P}_{x}$ was treated abstractly. To provide a concrete micro-foundation for this mechanism and to bridge our performative optimization framework with adversarial applications (such as the LLM safety guardrails discussed later), we now introduce an equivalent game-theoretic perspective. We present a formal model in which the probability measure governing uncertainty is endogenously generated by the best response of a homogeneous population of agents, and the optimizer responds to this induced distribution. Under this lens, the performative solution naturally coincides with a Nash fixed point.

\subsection{Agent's Problem with Baseline and Endogenous Change of Measure}
Let $X \subset \mathbb{R}^{d}$ be a nonempty compact set of decisions and $\mathcal{A} \subset \mathbb{R}^{k}$ be a nonempty compact set of agent actions. Let $(\Xi, \mathcal{F})$ be a measurable outcome space. We introduce a baseline probability space $(\Xi, \mathcal{F}, v_{0})$, where $v_{0}$ represents the distribution of outcomes generated by a reference population behavior. This baseline corresponds to a nominal action $a_{0} \in \mathcal{A}$, which can be interpreted as the default or status-quo level of effort in the absence of strategic adaptation. The measure $v_{0}$ therefore encodes the population profile under baseline behavior and serves as the reference point for both evaluation and deviation.

\subsubsection{Agent utility relative to baseline} 
The agent evaluates actions relative to the baseline population behavior captured by the reference distribution $v_{0}$. Formally, let $U: \mathcal{A} \times X \times \mathcal{P}(\Xi) \rightarrow \mathbb{R}$ be a utility function, and write $U\left(a ; x, v_{0}\right)$ to emphasize that the evaluation of an action $a$ depends on both the decision $x$ and the baseline distribution $v_{0}$.
A general and flexible specification is obtained by expressing the utility in terms of the distribution $v_{a}$ induced by action $a$. Specifically, suppose that for each $a \in \mathcal{A}$, the induced distribution $v_{a}$ satisfies $v_{a} \ll v_{0}$ with likelihood ratio $\frac{d v_{a}}{d v_{0}}(\xi)=L(a, \xi)$. 
Then the utility can be written in the form
\begin{equation}\label{eq:utility_general}
U\left(a ; x, v_{0}\right)=\int_{\Xi} r(x, \xi) d v_{a}(\xi)-C\left(v_{a} \| v_{0}\right) 
\end{equation}
where $r: X \times \Xi \rightarrow \mathbb{R}$ is a measurable reward function and $C\left(v_{a} \| v_{0}\right)$ is a nonnegative functional that measures the cost of deviating from the baseline distribution $v_{0}$.

The first term represents the expected performance under the distribution induced by the agent's action, while the second term penalizes deviations from the baseline population behavior. The functional $C$ may encode effort costs, risk sensitivity, or statistical divergence from the baseline.

\subsubsection{KL-regularized utility as a canonical example} 
A particularly important and tractable choice is obtained by taking $C$ to be the Kullback-Leibler divergence. In this case, the utility becomes $U\left(a ; x, v_{0}\right)=\int_{\Xi} r(x, \xi) d v_{a}(\xi)-\tau D_{\mathrm{KL}}\left(v_{a} \| v_{0}\right)$,
where $\tau>0$ is a regularization parameter and $D_{\mathrm{KL}}\left(v_{a} \| v_{0}\right)=\int_{\Xi} \log \left(\frac{d v_{a}}{d v_{0}}(\xi)\right) d v_{a}(\xi).$
Using the likelihood ratio representation, the utility admits the equivalent expression $U\left(a ; x, v_{0}\right)=\int_{\Xi} r(x, \xi) L(a, \xi) d v_{0}(\xi)-\tau \int_{\Xi} L(a, \xi) \log L(a, \xi) d v_{0}(\xi).$
Under this specification, the agent selects an action $a$ that induces a probability measure $v_{a}$ balancing two competing effects: improving expected performance under $r(x, \xi)$ and remaining close to the baseline distribution $v_{0}$. The resulting induced measure can be interpreted as an exponentially tilted version of $v_{0}$, with the parameter $\tau$ governing the strength of the deviation from the baseline.

\subsubsection{Best-response correspondence} 
For each $x \in X$, the agent solves $\mathcal{B}(x):=\arg \max _{a \in \mathcal{A}} U\left(a ; x, v_{0}\right) .$ We assume that $\mathcal{B}(x)$ is nonempty for all $x \in X$ and admits a measurable selection $b: X \rightarrow \mathcal{A} \text{ with } b(x) \in \mathcal{B}(x).$
The mapping $b$ describes how the agent adjusts effort in response to the decision $x$, taking the baseline population behavior as reference.

\subsubsection{Endogenous change of measure} 
The key modeling step is that the agent's action does not merely select an outcome, but alters the probability law itself. For each action $a \in \mathcal{A}$, we associate a probability measure $v_{a}$ on $(\Xi, \mathcal{F})$ satisfying $v_{a} \ll v_{0}, \quad \frac{d v_{a}}{d v_{0}}(\xi)=L(a, \xi),$ where $L: \mathcal{A} \times \Xi \rightarrow \mathbb{R}_{+}$ is a measurable likelihood ratio such that $\int_{\Xi} L(a, \xi) d v_{0}(\xi)=1 \text { for all } a \in \mathcal{A}.$
Thus, each action $a$ induces a tilting of the baseline measure $v_{0}$. The baseline corresponds to $a_{0}$, for which typically $L\left(a_{0}, \xi\right) \equiv 1$ and hence $v_{a_{0}}:=v_{0}$.

\subsubsection{Induced distribution} 
Given a decision $x \in X$, the agent responds with $b(x)$, and the resulting distribution of outcomes is defined by $\mathbb{P}_{x}:=v_{b(x)}.$ Equivalently, for any bounded measurable function $\varphi: \Xi \rightarrow \mathbb{R}$, $\int_{\Xi} \varphi(\xi) d \mathbb{P}_{x}(\xi)=\int_{\Xi} \varphi(\xi) L(b(x), \xi) d v_{0}(\xi).$
The above construction makes explicit that the dependence of the distribution on the decision $x$ is entirely mediated by the agent's best response. Specifically, we have the composition $x \longmapsto b(x) \longmapsto v_{b(x)} \longmapsto \mathbb{P}_{x}.$ The optimizer does not directly manipulate the probability law. Instead, it selects $x$, which changes the agent's incentives, leading to a new action $b(x)$ that induces a change of measure from $v_{0}$ to $\mathbb{P}_{x}$.

\subsection{Optimizer's Problem}
Given an action $a \in \mathcal{A}$ chosen by the agent, a probability measure $v_{a}$ on $(\Xi, \mathcal{F})$ is induced. The optimizer evaluates feasibility under the distribution generated by the agent's action $a$. For each $a \in \mathcal{A}$, define the feasible set $\tilde{\mathcal{X}}_{\varepsilon}^{\gamma}(a):=\left\{y \in X: v_{a}(g(y, \xi) \leq \gamma) \geq 1-\varepsilon\right\} . $
Equivalently, the feasibility condition can be written as $\int_{\Xi} \mathbf{1}\{g(y, \xi) \leq \gamma\} d v_{a}(\xi) \geq 1-\varepsilon.$
If the likelihood ratio representation is used, i.e., $d v_{a} / d v_{0}=L(a, \xi)$, then the constraint admits the form $\int_{\Xi} \mathbf{1}\{g(y, \xi) \leq \gamma\} L(a, \xi) d v_{0}(\xi) \geq 1-\varepsilon.$ The set $\tilde{\mathcal{X}}_{\varepsilon}^{\gamma}(a)$ consists of all decisions $y$ that satisfy the chance constraint under the distribution induced by the agent's action $a$. Hence, feasibility depends on $a$ through the mapping $a \mapsto v_{a}$. Define the optimizer best-response correspondence $\mathcal{R}(a):=\arg \min _{y \in \tilde{\mathcal{X}}_{\varepsilon}^{\gamma}(a)} f(y).$
For each action $a \in \mathcal{A}$, the set $\mathcal{R}(a)$ consists of all decisions that minimize the objective function subject to feasibility under the distribution $v_{a}$ induced by $a$.

\subsection{Nash Equilibrium}
We consider a game between the agent and the optimizer. The agent chooses an action $a \in \mathcal{A}$ in response to $x \in X$, while the optimizer chooses a decision $x \in X$ in response to $a \in \mathcal{A}$.

\begin{definition}\label{def:nash_equilibrium}
Consider the game $\mathcal{G}=(\{O, A\},\{X, \mathcal{A}\},\{U, f\})$, where the optimizer $O$ selects a decision $x \in X$ and the agent $A$ selects an action $a \in \mathcal{A}$. The agent evaluates actions according to the utility function $U\left(a ; x, v_{0}\right)$, and its best-response correspondence is given by $\mathcal{B}(x)=\arg \max _{a \in \mathcal{A}} U\left(a ; x, v_{0}\right)$. Each action $a \in \mathcal{A}$ induces a probability measure $v_{a}$ on $(\Xi, \mathcal{F})$, and the optimizer evaluates decisions under this induced distribution.

Given $a \in \mathcal{A}$, the optimizer solves $\min _{x \in X}\left\{f(x): v_{a}(g(x, \xi) \leq \gamma) \geq 1-\varepsilon\right\}$, with corresponding best-response correspondence $\mathcal{R}(a)=\arg \min _{y \in \tilde{\mathcal{X}}_{\varepsilon}^{\gamma}(a)} f(y)$.
A pair $\left(x^{\star}, a^{\star}\right) \in X \times \mathcal{A}$ is called a Nash equilibrium if
\begin{equation}\label{eq:nash_agent}
a^{\star} \in \mathcal{B}\left(x^{\star}\right), 
\end{equation}
and
\begin{equation}\label{eq:nash_optimizer}
x^{\star} \in \mathcal{R}\left(a^{\star}\right) . 
\end{equation}
\end{definition}

The Nash equilibrium defined above is equivalent to the performative solution characterized previously. Recall that the performative best-response operator is defined by $\Phi(x):=\arg \min _{y \in \mathcal{X}_{\varepsilon}^{\gamma}(x)} f(y) $, where the feasible set $\mathcal{X}_{\varepsilon}^{\gamma}(x)$ is defined using the distribution induced by the agent's best response $b(x)$.

\begin{proposition}\label{prop:equivalence}
A pair $(x^{\star}, a^{\star})$ is a Nash equilibrium in the sense of \eqref{eq:nash_agent} and \eqref{eq:nash_optimizer} if and only if
\begin{equation}\label{eq:performative_fixed_point}
x^{\star} \in \Phi\left(x^{\star}\right) .
\end{equation}
\end{proposition}

\begin{proof}
Suppose $(x^{\star}, a^{\star})$ is a Nash equilibrium. Then $a^{\star} \in \mathcal{B}\left(x^{\star}\right)$ and $x^{\star} \in \mathcal{R}\left(a^{\star}\right)$. By definition of $\mathcal{R}\left(a^{\star}\right)$, we have $x^{\star} \in \arg \min _{y \in \tilde{\mathcal{X}}_{\varepsilon}^{\gamma}\left(a^{\star}\right)} f(y)$, where the feasible set is evaluated under the distribution $v_{a^{\star}}$.

Since $a^{\star} \in \mathcal{B}\left(x^{\star}\right)$, the distribution induced by $x^{\star}$ satisfies $\mathbb{P}_{x^{\star}}=v_{b\left(x^{\star}\right)}=v_{a^{\star}}$. Therefore, for every $y \in X$, $v_{a^{\star}}(g(y, \xi) \leq \gamma)=\mathbb{P}_{x^{\star}}(g(y, \xi) \leq \gamma)$, which implies $\tilde{\mathcal{X}}_{\varepsilon}^{\gamma}\left(a^{\star}\right)=\mathcal{X}_{\varepsilon}^{\gamma}\left(x^{\star}\right)$. It follows that $x^{\star} \in \arg \min _{y \in \mathcal{X}_{\varepsilon}^{\gamma}\left(x^{\star}\right)} f(y)=\Phi\left(x^{\star}\right)$, and hence $x^{\star} \in \Phi\left(x^{\star}\right)$.

Conversely, suppose $x^{\star} \in \Phi\left(x^{\star}\right)$. By definition, $x^{\star} \in \arg \min_{y \in \mathcal{X}_{\varepsilon}^{\gamma}\left(x^{\star}\right)} f(y)$. Let $a^{\star} \in \mathcal{B}\left(x^{\star}\right)$. Then the induced distribution satisfies $P_{x^{\star}}=v_{a^{\star}}$, and the same argument as above yields $\mathcal{X}_{\varepsilon}^{\gamma}\left(x^{\star}\right)=\tilde{\mathcal{X}}_{\varepsilon}^{\gamma}\left(a^{\star}\right)$. Therefore, $x^{\star} \in \arg \min _{y \in \tilde{\mathcal{X}}_{\varepsilon}^{\gamma}\left(a^{\star}\right)} f(y)=\mathcal{R}\left(a^{\star}\right)$, and thus $(x^{\star}, a^{\star})$ satisfies \eqref{eq:nash_agent} and \eqref{eq:nash_optimizer}.
\end{proof}

\section{Case Study: Performative Evasion Detection in LLM Safety}

Modern AI security systems, such as guardrails for large language models (LLMs) and content moderation platforms, operate in a continuous cat-and-mouse dynamic. Designers deploy detection mechanisms to block malicious inputs, and in response, attackers continuously adapt through obfuscation, prompt mutations, or evasion strategies. This interaction is inherently performative: the deployed defense does not merely react to adversarial behavior but actively shapes it, causing the underlying data distribution to evolve as a function of the defense itself.

In practice, modeling these adaptive adversaries explicitly is notoriously difficult due to their heterogeneous and emerging strategies. This fragility motivates a model-free, data-driven design paradigm. Rather than attempting to parameterize adversarial behavior, defenders can rely on an iterative, closed-loop process—collecting realized evasion samples induced by the current defense to refine the next update. This performative scenario approach provides a principled way to design robust defenses against unknown, strategically evolving adversaries.

\subsection{Performative SVM for LLM Jailbreak Detection}
We now illustrate the performative chance-constrained framework using an emerging problem in AI safety: deploying robust guardrails against LLM jailbreaks. Viewed through the game-theoretic lens established in Section 4, this scenario naturally models a game where the defender (the optimizer) deploys a classifier to detect malicious prompt embeddings, while the strategic attacker (the agent) adversarially alters the prompt distribution to evade detection. 
In real-world deployments, to meet strict latency and computational cost constraints, LLM systems often employ a multi-tiered architecture \cite{chen2023frugalgpt}. The first layer—a low-latency ``fast guardrail''—typically evaluates user prompts by operating directly on their high-dimensional vector embeddings \cite{rebedea2023nemo}. Furthermore, recent advances in representation engineering demonstrate that complex safety concepts, such as harmfulness and malicious intent, are practically linearly separable in these embedding spaces \cite{zou2023representation}. Motivated by this industry standard, let $X \subset \mathbb{R}^{d}$ be a compact set of classifier parameters $w \in \mathbb{R}^{d}$ representing such a linear safety filter deployed in an embedding space. Let $(\xi, y) \in \mathbb{R}^{d} \times\{-1,+1\}$ denote a random prompt embedding and its underlying intent label, where $y=-1$ indicates a malicious/jailbreak intent and $y=+1$ indicates a benign query.

A central feature of this security setting is that the data distribution depends on the deployed guardrail, and we write $(\xi, y) \sim \mathbb{P}_{w}$. This dependence arises naturally due to the presence of strategic agents (e.g., red-teamers or attackers). When a guardrail $w$ successfully blocks explicit malicious queries, attackers utilize LLMs to semantically mutate their prompts—rewriting them as hypothetical or academic inquiries—to bypass the filter while retaining the core malicious intent. Thus, the deployed model directly alters the observed data distribution $\mathbb{P}_{w}$.

We now provide a microfoundation for this dependence using a strategic evasion model. Let $\xi_{0} \in \mathbb{R}^{d}$ denote a baseline feature vector of an explicit malicious prompt drawn from a fixed distribution $\mathbb{P}_{0}$. After observing the deployed guardrail $w$, an attacker chooses a mutated feature vector $\xi$ by solving $ \max_{\xi \in \mathbb{R}^{d}} y\langle w, \xi\rangle - c(\|\xi-\xi_{0}\|), $
where $c: \mathbb{R}_{+} \rightarrow \mathbb{R}_{+}$ is convex, increasing, and grows superlinearly. This explicitly instantiates the agent's utility maximization problem established in the game-theoretic framework of Section 4. Specifically, the utility balances the expected evasion success (which acts as the agent's reward) against the semantic distortion cost (the penalty for deviating too far from the original intent).

The first-order optimality condition implies that the displacement $\Delta:=\xi-\xi_{0}$ satisfies $ y w = c^{\prime}(\|\Delta\|) \frac{\Delta}{\|\Delta\|}, $
so that $\Delta$ is aligned with $y w$ and its magnitude is determined by the inverse of $c^{\prime}$. Consequently, the optimal strategic response can be written as $ \xi = \xi_{0} + y \alpha(\|w\|) w, $
where $\alpha: \mathbb{R}_{+} \rightarrow \mathbb{R}_{+}$ is an increasing but bounded function (e.g., $\alpha(r)=(\lambda+\kappa r)^{-1}$ for constants $\lambda, \kappa>0$). This captures a saturating response: as $\|w\|$ grows, the marginal benefit of further semantic manipulation diminishes. The induced distribution $\mathbb{P}_{w}$ is the pushforward of $\mathbb{P}_{0}$ under this transformation.

Define the margin-based safety constraint function $g(w,(\xi, y)):=1-y\langle w, \xi\rangle$. The condition $g(w,(\xi, y)) \leq 0$ corresponds to a correct classification (blocking malicious prompts and allowing benign ones) with a margin of at least one. For a tolerance level $\varepsilon \in(0,1)$, we impose the chance constraint $ \mathbb{P}_{w}(y\langle w, \xi\rangle \geq 1) \geq 1-\varepsilon. $
Under the strategic response model, the margin of the mutated prompt evaluated by the current guardrail can be expressed as $y\langle w, \xi\rangle = y\langle w, \xi_{0}\rangle + \alpha(\|w\|)\|w\|^{2}$. Thus, the chance constraint becomes $ \mathbb{P}_{0}\left(y\langle w, \xi_{0}\rangle \geq 1-\alpha(\|w\|)\|w\|^{2}\right) \geq 1-\varepsilon. $ This representation explicitly demonstrates the dual dependence in our theoretical framework: the decision variable $w$ acts as the evaluator in the constraint function $g$, while simultaneously acting as the environment trigger that shapes the data distribution $\mathbb{P}_{w}$. 
We consider the strongly convex regularized objective $f(w)=\frac{1}{2}\|w\|^{2}$. The performative chance-constrained guardrail problem is formulated as 
\[ \min_{w \in X} \frac{1}{2}\|w\|^{2} \quad \text{subject to} \quad \mathbb{P}_{w}(y\langle w, \xi\rangle \geq 1) \geq 1-\varepsilon. \]

For a reference guardrail $w \in X$, the best-response mapping is $\phi(w)=\arg \min_{\tilde{w} \in \mathcal{X}_{\varepsilon}(w)} \frac{1}{2}\|\tilde{w}\|^{2}$. To obtain a tractable approximation, we adopt the scenario approach. For a fixed $w \in X$, we draw i.i.d. prompt samples $\left(\xi^{(1)}, y^{(1)}\right), \ldots,\left(\xi^{(N)}, y^{(N)}\right) \sim \mathbb{P}_{w}$ from the LLM via simulated API interactions. The chance constraint is replaced by empirical constraints $y^{(i)}\langle\tilde{w}, \xi^{(i)}\rangle \geq 1$ for all $i=1, \ldots, N$. Denoting its unique solution by $\phi_{N}(w)$, the induced stochastic iteration is given by $ w_{t+1}=\phi_{N_{t}}\left(w_{t} ;\left\{\left(\xi_{t}^{(i)}, y_{t}^{(i)}\right)\right\}_{i=1}^{N_{t}}\right), $
defining a performative learning dynamics where the guardrail is repeatedly updated using adversarial prompts generated in response to its own deployment.

\begin{figure*}[t]
  \centering
  \includegraphics[width=0.9\textwidth]{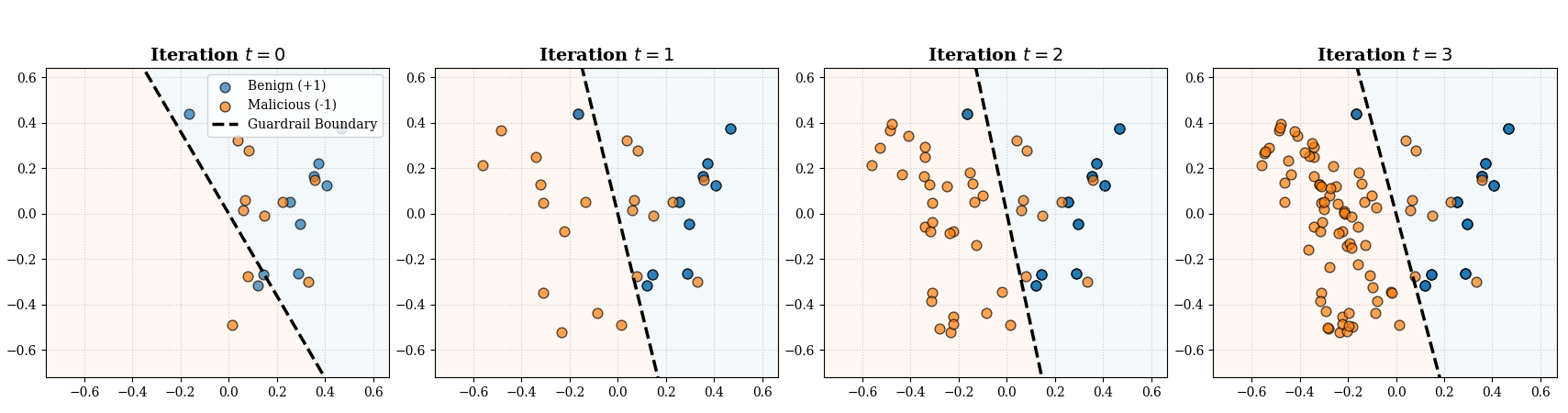}
  \caption{Evolution of the performative guardrail and the induced semantic shift in the embedding space. The dashed line denotes the decision boundary $w_t$, and the data points represent benign (blue) and mutated malicious (orange) prompts.}
  \label{fig:performativity_evolution}
\end{figure*}

\subsection{Numerical Illustration}
We illustrate this performative scenario-based framework using a linear SVM guardrail trained on text embeddings under strategic prompt mutation.

Figure \ref{fig:performativity_evolution} shows the evolution of the classifier and the induced data distribution across iterations. Several key phenomena emerge from the results. First, the data distribution physically evolves in response to the guardrail. As the classifier successfully blocks explicit threats at $t=0$, attackers mutate their features, leading to a noticeable semantic shift (the upward-right migration of the orange cluster) at subsequent iterations. This perfectly visualizes the performative nature of the system.

Second, as shown in Figure \ref{fig:convergence_and_growth} (Left), the guardrail weights and the adversarial distribution co-evolve toward a stable configuration. The high-dimensional iterates swiftly approach a fixed point, consistent with our theoretical characterization of performative solutions as fixed points of the best-response mapping. Figure \ref{fig:convergence_and_growth} (Right) illustrates the increasing sample size $N_{t}$ which reduces the stochastic error in the scenario approximation. Early iterations use small sample sizes for rapid but noisy updates, while later iterations accumulate a larger pool of mutated scenarios, yielding more accurate approximations. This matches the theoretical requirement that $N_{t}$ must grow logarithmically to ensure almost sure convergence.

\begin{figure}[htbp]
  \centering
  \includegraphics[width=1\linewidth]{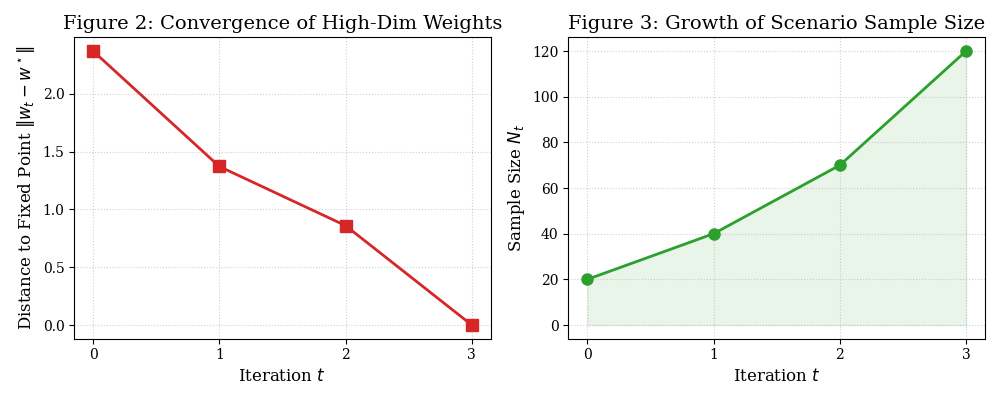}
  \caption{(Left) Convergence of the stochastic iteration to the performative fixed point in the embedding space. (Right) Logarithmic growth of the sample size $N_{t}$ used in the scenario optimization.}
  \label{fig:convergence_and_growth}
  \vspace{-3mm}
\end{figure}

\section{CONCLUSIONS}

In this paper, we introduced a novel framework for performative scenario optimization to address chance-constrained problems in decision-dependent environments. By formally modeling the feedback loop between the decision-maker and the data-generating process, we characterized performative solutions as self-consistent equilibria and established their existence via Kakutani’s fixed-point theorem. To overcome the intractability of unknown induced distributions, we developed a scenario-based approximation method. We proved that a stochastic fixed-point iteration, equipped with a logarithmic sample size schedule, converges almost surely to the unique performative solution. The theoretical findings were empirically validated through an LLM jailbreaking problem, confirming the stable co-evolution of the classifier and the data distribution.

\addtolength{\textheight}{-12cm}   % This command serves to balance the column lengths

%%%%%%%%%%%%%%%%%%%%%%%%%%%%%%%%%%%%%%%%%%%%%%%%%%%%%%%%%%%%%%%%%%%%%%%%%%%%%%%%

\bibliographystyle{IEEEtran}
\bibliography{IEEEexample}

\end{document}